\documentclass[runningheads]{llncs}
\usepackage{graphicx}
\usepackage{amsmath}
\usepackage{amssymb}
\usepackage{listings}
\usepackage[]{algorithm2e}
\usepackage{stmaryrd}
\usepackage{todonotes}
\usepackage{cleveref}
\newcommand{\ap}{\mathrm{ap}}
\newcommand{\eventually}{\operatorname{\mathbf{F}}}
\newcommand{\false}{\mathrm{false}}
\newcommand{\globally}{\operatorname{\mathbf{G}}}
\newcommand{\until}{\mathrel{\mathbf{U}}}
\newcommand{\lnext}{\operatorname{\mathbf{X}}}
\newcommand{\rarrow}{\boldsymbol{\rightarrow}}
\newcommand{\theory}{\mathcal{T}}
\newcommand{\tdomain}{\mathbb{T}}
\newcommand{\bool}{\mathcal{B}}
\newcommand{\updf}{\mathbf{u}}
\newcommand{\updfs}{\mathbf{U}}
\newcommand{\rvar}{\mathbf{r}}
\newcommand{\rvars}{\mathbf{R}}
\newcommand{\ivar}{\mathbf{i}}
\newcommand{\ivars}{\mathbf{I}}
\newcommand{\quantelim}[1]{\texttt{quantelim}(#1)}
\newcommand{\isconsistent}[1]{\texttt{isconsistent}(#1)}
\newcommand{\sat}[1]{\texttt{sat}(#1)}
\newcommand{\true}{\mathrm{true}}
\newcommand{\bexpr}{E_\theory^\mathbb{B}(R \cup I)}
\newcommand{\texpr}{E_\theory^\tdomain(R \cup I)}
\newcommand{\nt}[1]{\langle\text{#1}\rangle}
\newcommand{\consistent}{\textsc{consistent}}
\newcommand{\inconsistent}{\textsc{inconsistent}}

\newcommand{\unrealizable}{\textsc{unrealizable}}
\newcommand{\vvar}{\mathbf{v}}

\begin{document}
\title{Reactive Synthesis Modulo Theories}
\subtitle{Using Abstraction Refinement}
%
%
\author{Benedikt Maderbacher \and Roderick Bloem}
\authorrunning{}
%
\institute{Graz University of Technology, Austria}
\maketitle              
\begin{abstract}
Reactive synthesis builds a system from a specification given as a temporal logic formula.
Traditionally, reactive synthesis is defined for systems with Boolean input and output variables.
Recently, new theories and techniques have been proposed to extend reactive synthesis to data domains, which are required for more sophisticated programs.
In particular, Temporal stream logic (TSL) \cite{finkbeiner2019a} extends LTL with state variables, updates, and uninterpreted functions and was created for use in synthesis.
We present a synthesis procedure for TSL(T), an extension of TSL with theories.
Synthesis is performed using a counter-example guided synthesis loop and an LTL synthesis procedure.
Our method translates TSL(T) specifications to LTL and extracts a system if synthesis is successful.
Otherwise, it analyzes the counterstrategy for inconsistencies with the theory. If the counterstrategy is theory-consistent, it proves that the specification in unrealizable. Otherwise, we add temporal assumptions and Boolean predicates to the TSL(T) specification and start the next iteration of the the loop.
We show that the synthesis problem for TSL(T) is undecidable. Nevertheless our method can successfully synthesize or show unrealizability of several non-Boolean examples. 
\end{abstract}

\section{Introduction}

Reactive synthesis \cite{DBLP:reference/mc/BloemCJ18} is the problem of automatically constructing a system from a specification.
The user provides a specification in temporal logic and the synthesis procedure constructs a system that satisfies it.
Traditionally this only works for systems with Boolean input and output variables.
However, real world system often use more sophisticated data like integers, reals, or structured data.
For finite domains it is possible to use bit-blasting to obtain an equivalent Boolean specification.
These will be hard for a human to read and the large number of variables make them  very challenging for a synthesis tool to solve.

In recent years multiple theories have been proposed to perform reactive synthesis with non Boolean inputs and outputs.
There have been decidability results for synthesis using register automata \cite{ehlers2014,khalimov2018,exibard2021} and variable automata \cite{faran2020}.

Our work builds on temporal stream logic (TSL). TSL, proposed by Finkbeiner et al.\ \cite{finkbeiner2019a}, uses a logic based on linear temporal logic (LTL) with state variables, uninterpreted functions and predicates, and update expressions. TSL allows for an elegant and efficient synthesis method that separates control from data.
However, the ability to specify how data is handled is limited because functions and predicates remain uninterpreted.
Finkbeiner et al. \cite{finkbeiner2021} describe an extension to TSL modulo theories, but consider only satisfiability and not synthesis.

In this paper we propose a synthesis algorithm for temporal stream logic modulo theories that can be applied to arbitrary decidable theories in which quantifier elimination is possible.
Let us consider a concrete example using the theory of linear integer arithmetic (LIA).

\begin{example}
We want to build a system with one integer state variable $x$ and one integer input $i$.
The objective is to keep the value of the state variable between $0$ and $100$.
At any time step the system can select one of two updates: increase or decrease $x$ by $i$, where $i$ is chosen by the environment in the interval $0 \leq i < 5$.
We assume that the initial state is any value inside the boundaries.
These requirements can be written as the TSL formula
\begin{multline*}
    \phi \triangleq (0 \leq x \land x < 100 \land \globally (0 \leq i \land i < 5)) \rarrow \\
\globally (0 \leq x \land x < 100 \land ([x \leftarrow x-i] \lor [x \leftarrow x+i])),    
\end{multline*}
where the propositions $[x \leftarrow x-i]$ and $[x \leftarrow x+i]$ describe updates to $x$.

A human programmer might write the following program satisfying the specification:
\begin{lstlisting}
while(true)
    i := receive()
    if (x-i>=0)
         x := x - i
    else
         x := x + i
\end{lstlisting}
Note how this program uses a condition that doesn't appear in the original specification.
In fact, it is impossible to write a correct system using only the predicates from the specification.
\end{example}

Inspired by this example we want our synthesis algorithm to function with expressions from theories and to find new predicates when necessary. Our algorithm is similar to the one in the original TSL synthesis paper \cite{finkbeiner2019a}.
The TSL specification is encoded into an LTL formula that contains a Boolean variable for each theory predicate in the TSL formula. These variables are seen as inputs, which means that the environment determines their truth values. As a result, realizability of the LTL formula implies realizability of the TSL formula, but not vice versa, because the environment can choose values for the variables that are not consistent with the theory. 
The LTL formula is then given to a propositional LTL synthesis tool \cite{meyer2018,schewe2007}. If Boolean synthesis is successful we obtain a Boolean system that can be concretized into a system that operates on the original value domain.
If synthesis of the LTL formula is not successful, we get a Boolean counter strategy that we analyze for inconsistencies with respect to the theory.
An theory-consistent counter strategy means that the TSL specification is unrealizible.
If an inconsistency is found the counter strategy is spurious and new assumptions and possibly new predicates are generated and integrated into the TSL specification.

The procedure can be likened to a CEGAR loop \cite{clarke2000},  or to the DPLL(T) in which LTL synthesis plays the role of the propositional SAT solver and the consistency check is performed by the theory solver. The main difference is that in our case inconsistencies can span multiple time steps.

\Cref{fig:overview} shows an overview of our approach.
We will show that the synthesis problem for TSL modulo theories is undecidable and the process is thus not guaranteed to terminate.

\begin{figure}
    \centering
    \includegraphics[width=0.9\textwidth]{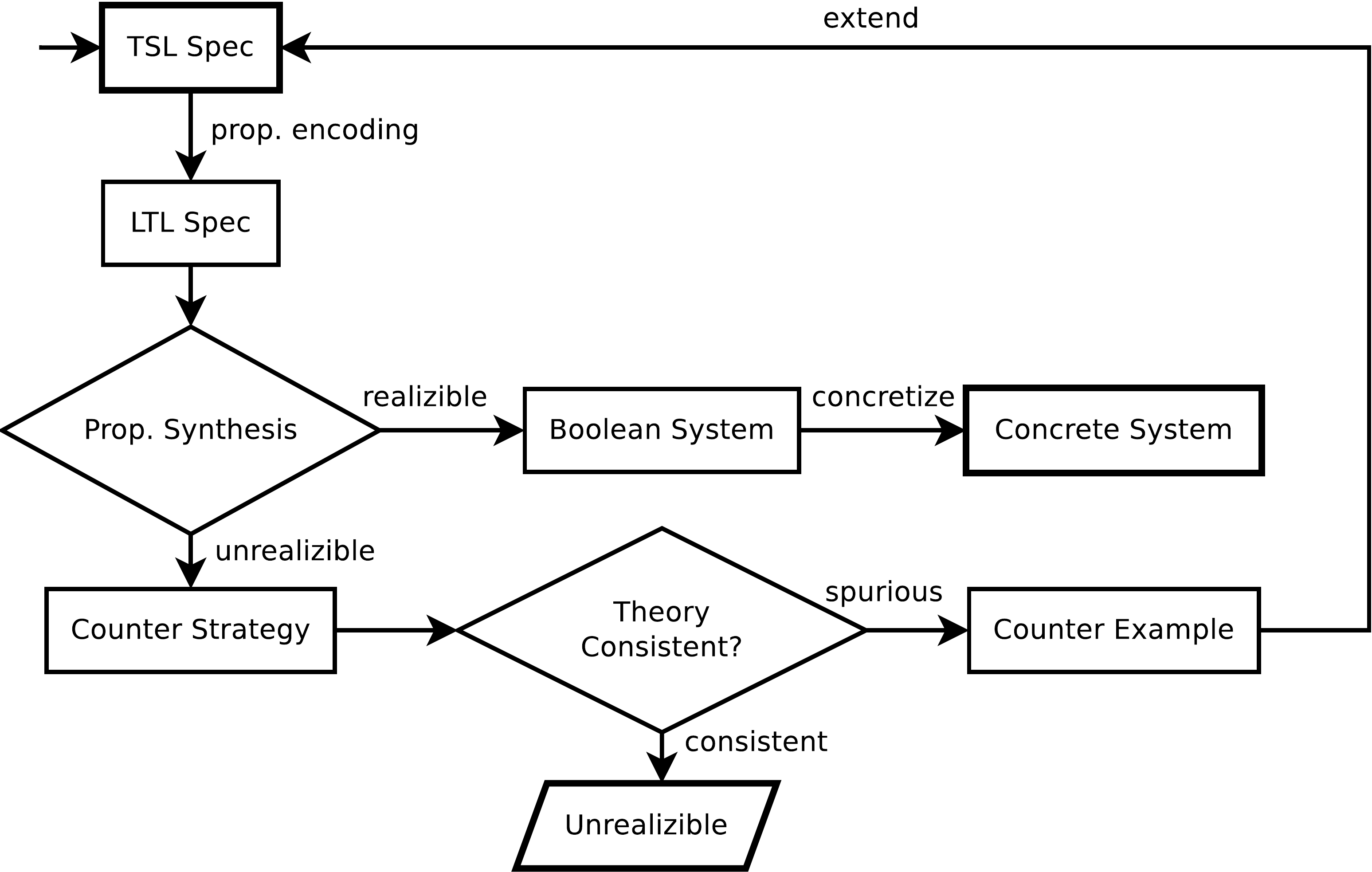}
    \caption{Overview of the synthesis procedure.}
    \label{fig:overview}
\end{figure}

The central part of our algorithm and main contribution of this paper is the theory consistency analysis of a counter strategy.
It uses an theory solver --- e.g., an SMT solver --- to locally analyze the states of transitions of the counter strategy to detect inconsistencies with respect to the theory. The assumptions it generates will contain new predicates where necessary.

The remaining paper is structured as follows:
\Cref{sec:def} contains required definitions and formalizes the synthesis problem for TSL modulo theory.
We describe the Boolean abstraction and the theory consistency analysis in detail in \cref{sec:abstr}.
The main synthesis procedure is described in \cref{sec:synthesis}.
An experimental evaluation was performed for multiple examples using the theories of linear integer arithmetic and linear real arithmetic (\cref{sec:experiments}). We conclude our work in \cref{sec:conclusion}.

\section{Preliminaries}
\label{sec:def}

We use Temporal Stream Logic (TSL) \cite{finkbeiner2019a}, with the addition of decidable theories.


\subsection{Theories and Updates}
 A theory $\mathcal{T}$ consists of a signature (symbols for constants, functions, and predicates) and a domain of values $\tdomain$. All symbols are assumed to have a fixed interpretation in the domain $\tdomain$. In the following we will use $E_\theory(V)$ to denote the set of expressions in $\theory$ with free variables that are a subset of $V$. The set $E_\theory(V)$ is partitioned into a set $E^{\tdomain}_\theory(V)$ of terms (denoting values in $\tdomain$) and a set $E^{\mathbb{B}}_\theory(V)$ of formulas (denoting truth values). 
 We assume that the theories used have decidable procedures for satisfiability checking and quantifier elimination. We assume that we are given a procedure \sat that returns true iff a formula $\phi$ is satisfiable and a function \quantelim that takes a formula and returns a theory-equivalent formula that does not contain quantifiers. 
 
\begin{example}
The signature of the theory of linear integer arithmetic (LIA) can be defined in Backus-Naur form:
\begin{align*}
    \nt{const} &:= 0 \mid 1 \mid 2 \mid \ldots\\
    \nt{var} &:= x \mid y \mid z \mid \ldots\\
    \nt{term} &:= \nt{const} \mid \nt{var} \mid - \nt{int} \mid \nt{int} + \nt{int} \mid \nt{const} * \nt{int} \\
    \nt{formula} &:= \nt{int}= \nt{int} \mid \nt{int} \leq \nt{int} \mid \forall  \nt{var}. \nt{formula}
\end{align*}

\end{example}

\paragraph{Updates}
In the following, we will use sets of state variables $R$ and input $I$ variables. The new values of the state variables are determined using update functions. An update function $\updf$ defines an update $\updf(r) \in \texpr$ for each $r \in R$. The set of all update functions is denoted by $\updfs$.


We introduce the notations $\rvars \triangleq R \to \tdomain$ and $\ivars \triangleq I \to \tdomain$ for valuations of variables. We write $R/\rvar$ ($I/\ivar$) to denote the replacement of all variables in $R$ ($I$, resp.) by their corresponding values in $\rvar$ ($\ivar$, resp.). With slight abuse of notation, we identify $e[R/\rvar, I/\ivar]$ with the corresponding value in the domain.

To apply an update function $\updf \in \updfs$ to valuations $\rvar \in \rvars$ and $\ivar \in \ivars$  we write $\updf[\rvar,\ivar]$ which is defined as $\updf[\rvar,\ivar](r) = \updf(r)[R/\rvars, I/\ivars]$ for each $r$.

\subsection{Temporal Stream Logic Modulo Theories}
TSL(T) is based on linear temporal logic, but instead of Boolean variables it uses updates and Boolean theory expressions.
The grammar for TSL(T) formulas is
\begin{align*}
    \nt{ap} &:= \texpr\\
    \nt{theory} &:= \bexpr\\
    \nt{bconst} &:= \true \mid \false\\
    \nt{upd} &:= [ \nt{var} \leftarrow \nt{theory}]\\
    \nt{tsl} &:= \nt{ap} \mid \nt{upd} \mid \nt{bconst} \mid \neg \nt{tsl} \mid  \nt{tsl} \land \nt{tsl} \mid \nt{tsl} \until \nt{tsl} \mid \lnext \nt{tsl}.
\end{align*}

The semantics of TSL(T) are defined with respect to a trace of inputs and state variable configurations $\rho \in (\ivars \times \rvars)^\omega$ as follows. We assume that $\rho = \rho_0, \rho_1, \dots$ and that $\rho_j = (\rvar_j, \ivar_j)$ and we define
\begin{align*}
\rho \models p                 &\text{ iff } \rho_0 \models p \text{ for $p \in \nt{ap}$}\\
\rho \models [r \leftarrow e] &\text{ iff } \rvar_1= e(R/\rvar_0,I/\ivar_0) \\
\rho \models \true & \\
\rho \not\models \false & \\
\rho \models \neg \phi  &\text{ iff } \rho \not\models \phi\\
\rho \models \phi \wedge \psi  &\text{ iff } \rho \models \phi \text{ and } \rho \models \psi\\
\rho \models \phi \until \psi  &\text{ iff } \exists j. \rho_j, \rho_{j+1},\dots \models \psi \text{ and } \forall i<j. \rho_i,\rho_{i+1},\dots \models \phi \\
\rho \models \lnext \phi  &\text{ iff } \rho_1, \rho_2,\dots  \models \phi\\
\end{align*}

The unary temporal operators eventually ($\eventually$) and globally ($\globally$) can be added using their usual definitions: $\eventually \varphi \equiv \true \until \varphi$ and $\globally \varphi \equiv \neg \eventually \neg \varphi$.

\subsection{Theory Mealy and Moore Machines}
Theory Mealy machines are state machines with inputs and  register variables that range over the theory domain. The updates to the state variables and the registers are restricted by a set of predicates and the selected update functions.

A \emph{Theory Mealy Machine} $M_{\theory} = (Q, q_0, P, \rvar_0, \delta, \mu)$ consists of a finite set of states $Q$, an initial state $q_0 \in Q$, a finite set of predicates $P \subseteq \bexpr$, an initial valuation $\rvar_0 \in \rvars$, a transition function $\delta \in (Q \times 2^P) \to Q$ and a update selection function $\mu \in (Q \times 2^P) \to \updfs$.

For a given valuation $\vvar = (\rvar, \ivar)$, let $P_{\vvar}\subseteq P$ be the subset of predicates that is true in $\vvar$: $P_{\vvar} = \{p \in P \mid \vvar \models p\}$.

A run $\sigma$ of a theory Mealy machine induced by a sequence of input valuation $\bar{\ivar} = \ivar_0, \ivar_1,\dots \in \ivars^\omega$ is an infinite sequence of states $Q$ and valuations $\rvars$ $(q_0,\rvar_0), (q_1,\rvar_1), \ldots$. Any two consecutive configurations $(q_i,\rvar_i)$ and $(q_{i+1}, \rvar_{i+1})$ must be related by $q_{i+1} = \delta(q_i, P_{(\rvar_i,\ivar_i)})$ and $\rvar_{i+1} = \updf_i[\rvar_i, \ivar_i]$ where $\updf_i = \mu(q_i, P_{(\rvar_i, \ivar_i)})$.

A Theory Mealy machine $M_\theory$ realizes a TSL(T) formula $\phi$ if for all inputs sequences $\bar{\ivar} \in \ivars^\omega$ the resulting trace $\rho \equiv (\bar{\ivar}, \bar{\rvar})$ satisfies $\phi$.

We also define Theory Moore machines, which read the updates produced by a Mealy machine and produce the inputs read by a Mealy machine. Intuitively, Mealy machines are used to show realizability of a TSL(T) specification, while Moore machines are used to show their unrealizability.
A \emph{Theory Moore Machine} $M_{\theory} = (Q, q_0, P, \rvar_0, \delta, \iota)$ consists of a finite set of states $Q$, an initial state $q_0 \in Q$, a finite set of predicates $P \subseteq \bexpr$, an initial valuation $\rvar_0 \in \rvars$, and a transition function $\delta \in (Q \times \updfs) \to Q$, and $\iota: Q \times \rvars \rightarrow \ivars$ is the output function.  

A run $\sigma$ of a theory Moore machine induced by an infinite sequence of update functions $\bar{\updf} \in \updfs^\omega$ is a sequence of states and valuations $(q_0,\rvar_0), (q_1,\rvar_1), \ldots$. Any two consecutive entries $(q_i,\rvar_i)$ and $(q_{i+1}, \rvar_{i+1})$ must be related by $q_{i+1} = \delta(q_i, \updf_i)$ and $\rvar_{i+1} = \updf_i[\rvar_i, \iota(q_i)]$.

\section{Boolean Abstraction}
\label{sec:abstr}

\subsection{Propositional Encoding of TSL(T)}
In this section, we describe the propositional encoding of TSL(T), which closely follows that of Finkbeiner et al.\ \cite{finkbeiner2019a}.


A TSL(T) formula $\phi$ is encoded to an LTL formula $\phi_\bool$. Formula $\phi_\bool$ is obtained by replacing each update $u$  in $\phi$ by a Boolean output variables $p_u$ and  each atomic proposition $\ap$ by a Boolean input variable $p_{\ap}$.
Additionally, the formula  ensures that for each variable exactly one update is active at any point in time.
This results in:
$$
\phi_\bool \triangleq \globally(\bigwedge_r \bigvee_i( p_{[r \leftarrow e_i]} \land \bigwedge_{j\neq i} \neg p_{[r \leftarrow e_j]})) \land \phi[ap / p_{ap}, \ldots, u / p_u, \ldots].
$$


\begin{example}
The TSL(T) formula $\phi \triangleq 0 \leq x \land x<5 \rightarrow \globally(0 \leq x \land x<5 \land ([x \leftarrow x-1] \lor [x \leftarrow x+1]))$ is encoded as the LTL formula 
\begin{multline*}
\phi_\bool \triangleq  \globally (p_{[x \leftarrow x-1]} \wedge \neg \ p_{[x \leftarrow x+1]} \vee p_{[x \leftarrow x+1]} \wedge \neg p_{[x \leftarrow x-1]}) \land\\
(p_{0 \leq x} \land p_{x<5} \rightarrow \globally(p_{0 \leq x} \land p_{x<5} \land (p_{[x \leftarrow x-1]} \lor p_{[x \leftarrow x+1]}))),
\end{multline*}
where $p_{0 \leq x}$ and $p_{x<5}$ are input variables and $p_{[x \leftarrow x-1]}$ and $p_{[x \leftarrow x+1]}$ are output variables.
\end{example}

\subsection{Boolean Mealy and Moore Machines}

We will now define Boolean  Mealy and Moore machines.
These are the abstract versions of theory machines defined above that can be checked against  propositionally encoded TSL(T) specifications.
They allow us to link Boolean systems realizing a PTSL formula to theory systems.
We also study when a Boolean system is consistent with the theory it abstracts over.

The Boolean abstraction is based on a finite set of predicates $P \subset \bexpr$.

A \emph{Boolean Mealy machine} is a tuple $(Q, P, q_0, \delta_\bool, \mu_\bool)$, where
$Q$ is a set of states, $P$ is a set of predicates, $q_0 \in Q$ is the initial state, $\delta_\bool \in Q \times 2^P \to Q$ is the transition function, and $\mu_\bool \in Q \times 2^P \to \updfs$ is the update selection function.

A trace $\sigma_\bool$ of a Boolean Mealy machine induced by a sequence $\bar{\varphi} \in (2^P)^\omega$ is an infinite alternating sequence of states and updates ${(Q\times \updfs)}^\omega$, $(q_0^\bool, \updf_1^\bool, q_1^\bool, \updf_2^\bool, \ldots)$ where $q_{i+1} = \delta_\bool(q_i, \varphi_i)$ and $\updf_{i+1} = \mu_\bool(q_i,\varphi_i)$.

A Boolean Mealy machine $M_\bool$ is theory consistent with respect to theory $\theory$ iff there exists a theory machine $M_\theory$ such that all $v \models p_0$ are initial valuations and every run of $M_\theory$ is contained in $M_\bool$. More precisely this means that $Q$, $q_0$, and $U$ are the same for both machines. Additionally, $\delta_\bool$ must be compatible with $\delta$ and $\mu_\bool$ with $\mu$, i.e., $\forall q, \varphi, \rvar, \ivar \ldotp \rvar\|\ivar \models \varphi \Rightarrow (\delta(q,\rvar,\ivar)= \delta_\bool(q,\varphi) \land \mu(q,\rvar,\ivar)= \mu_\bool(q,\varphi))$.

A Boolean Moore machine $M_{\bool} = (Q, P, q_0, \delta, o)$ is an abstract version of a theory Moore machine $M_{\theory}$.
$Q$ is a set of states, $P$ is a set of predicates, $q_0 \in Q$ is the initial state,  $\delta \in Q \times \updfs \to Q$ is the transition function, and $o \in Q \to 2^P$ is the output function.
The set $\updfs$ can be encoded as a set of Boolean variables, with one variable for every update in $U$ the input then becomes a valuation of these Boolean variables which has to conform to an update function.
A run $\sigma_\bool$ of a Boolean Moore machine induced by a sequence $\bar{\updf} \in \updfs^\omega$ is an infinite sequence of states ${Q}^\omega$, $(q_0, q_1, \ldots)$ where $q_{i+1} = \delta(q_i, \updf_i)$.

A Boolean Moore machine $M_\bool$ is theory consistent with respect to theory $\theory$ iff there exists a theory machine $M_\theory$ such that $\rvar_0 \models o(q_0)$ and every run of $M_\theory$ is contained in $M_\bool$. More precisely this means that $Q$, $q_0$, $U$, and $\delta$ are the same for both machines and a run $\sigma_\theory \subset \sigma_\bool$ iff $\forall i \ldotp \rvar_i \models o(q_i)$.
In this case we write $M_\theory \sim M_\bool$.


\subsection{Theory Consistency Analysis}
When an environment strategy is deemed inconsistent the specification is extended in one of three ways: a new single state assumption, a new transition assumption, or a new predicate.
The counter strategy analysis is performed in two stages. The first checks for consistency of outputs in a single state while the second one checks consistency of transitions.



        

\begin{algorithm}[htb]
\SetKwProg{Def}{def}{:}{}
\Def{\isconsistent{$m$}}{
 \KwData{Boolean Moore machine $m = (Q, P, U, q_0, \delta, o)$}
 \KwResult{\consistent\ or (possibly) \inconsistent\ with additional assumptions.}
 \ForEach{$q \in \operatorname{reachable}(Q)$ \tcp*{Case 1}}{
    \If{$\neg \sat{o(q)}$}{
        \textbf{yield} \inconsistent, $\globally \neg o(q)$)\;
    }
 }

 \ForEach{$(q_i,\updf,q_j) \in \operatorname{reachable}(\delta)$ \tcp*{Case 2}}{
    \If{$\neg \sat{ o(q_i) \land \updf \land  o(q_j)' }$}{
        \textbf{yield} \inconsistent, $\globally ( o(q_i) \land \updf \rightarrow \lnext \neg o(q_j))$\;
    }
 }

 \ForEach{$(q_i,\updf,q_j) \in \operatorname{reachable}(\delta)$ \tcp*{Case 3}}{
    \If{$\sat{ \forall i'_0, \ldots, i'_m \ldotp o(q_i) \land  \updf  \land \neg  o(q_j)' }$}{
        $wp$ := weakest precondition ($ \updf $, $ o(q_j)'$)\;
        $\rho$ := \quantelim{$\exists i'_0, \ldots, i'_m \ldotp wp \land o(q_i)$}\;
        \textbf{yield} \inconsistent, $\globally(\neg \rho \land o(q_i) \land \updf \rightarrow \lnext \neg o(q_j))$\;
    }
 }
 \textbf{return} \consistent\;
 }
 \caption{Consistent with theory and using inputs.\label{alg:consistency:inputs}}
\end{algorithm}

\paragraph{State Consistency}
To check state consistency we look at every (reachable) state in the counter strategy and use an SMT solver to check if the output assignment is consistent with the theory. For example the two variables $p_{x\geq5}$ and $p_{x<0}$ cannot be true in the same state.
If such a problem is found we generate a new assumption that rules out this assignment in every state. In the previous example this would generate the assumption $\globally(\neg p_{x\geq5} \vee \neg p_{x<0})$. This process is similar to the interaction between a SAT solver and a theory solver in an SMT solver with lazy encoding.

\paragraph{Transition Consistency} Once all states produce consistent outputs and there still exists a counter strategy, we turn towards transitions. As of now there are no assumptions that link the state before an action was performed to the state afterwards. This step creates these assumptions where necessary and also finds new predicates if the existing ones are not sufficient.

We again use an SMT solver to perform this analysis. Let's look at the transition $\{p_{x\geq5}\} [x \leftarrow x+1] \{\neg p_{x\geq5}\}$. To check it the following SMT problem is generated $x\geq5 \land x'=x+1 \land \neg(x'\geq5)$, this is unsatisfiable and we can generate an assumption to eliminate it $\globally(p_{x\geq5} \land [x \leftarrow x+1] \rightarrow \lnext p_{x\geq5})$. 

Another case is that a transition is possible for some, but not all of values. For instance, the triple $\{p_{x<0}\} [x \leftarrow x+1] \{p_{x \geq 0}\}$ does not hold for all values of $x$. This shows that our current abstraction might not be precise enough to correctly describe this transition. In that case we calculate the weakest precondition of the post state given the updates of the transition. This gives us the new predicate $x\geq-1$. In case there are input variables we give the environment the benefit of doubt when checking if a transition is inconsistent. A transition is inconsistent if for all possible inputs at the next time step we do not satisfy the target state predicate. After computing the weakest precondition we use existential quantifier elimination to remove the future inputs from our predicate. In that case we are looking for states where transition is valid, which changes the quantifiers from forall to exists.


\begin{example}
\Cref{alg:consistency:inputs} checks consistency on a local level, the environment strategy can still be globally consistent if case 3 reports \inconsistent.
The following Boolean Moore machine has two transitions (orange) that are reported as \inconsistent by \Cref{alg:consistency:inputs} even though the machine is globally consistent.

\includegraphics[width=0.6\textwidth]{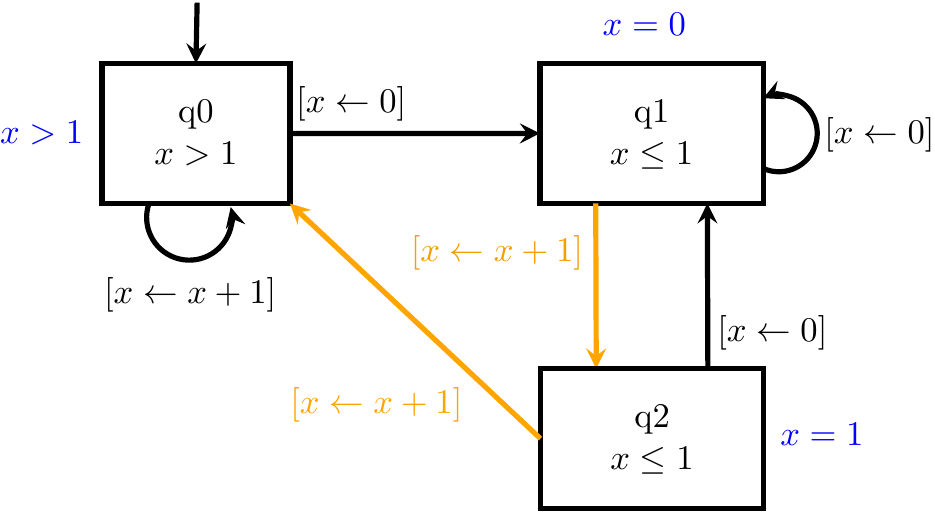}

The transition $(q1, [x \leftarrow x+1], q2)$ would be invalid for $x=1$ in $q1$, but for every execution $x=0$ in $q1$ and the problem does not appear. A similar situation occurs for the transition $(q2, [x \leftarrow x+1], q0)$, where $x$ will always be $1$ and the transition is only invalid for $x<1$.
The blue annotations show the possible values of $x$ in every state, demonstrating that all transitions are actually consistent.

The fact that we are over eager when reporting inconsistencies means that it will be harder to show unrealizability, but it does not affect our ability to find realizing systems.
\end{example}

\begin{lemma}
If \cref{alg:consistency:inputs} reports consistent for a Boolean Moore machine $M_\bool$ there exists a theory Moore machine $M_\theory$ such that $M_\theory \sim M_\bool$.
\label{lem:consistent}
\end{lemma}
\begin{proof}
Assuming the algorithm returns consistent.
Every output function is satisfiable, this includes the initial state which contains an initial value of the theory machine.
For every transition $(\varphi_i, \updf, \varphi_j)$ where there exists a model of $\varphi_i$ all of the models map to a model of $\varphi_j$, by the transition property checked by the algorithm.
Therefore by induction all paths starting in the initial state and only using transitions from $\delta$ contain a path of a theory machine.
A theory machine $M_\theory$ is included in $M_\bool$.
\qed
\end{proof}

\begin{lemma}
All assumptions $\psi$ added by \cref{alg:consistency:inputs} are tautological with respect to the theory.
$$\forall M_\theory \ldotp M_\theory \models \psi$$

\label{lem:tautogicalAssumptions}
\end{lemma}
\begin{proof}
The algorithm can produce three different types of assumptions $\psi_s, \psi_t, \psi_p$ corresponding to the three cases of the algorithm.

Let $\psi_s$ be $\globally \neg o(q)$ for an unsatisfiable $o(q)$.
$M_\theory$ must define a output valuation for every state, because $o(q)$ is empty no state in any $M_\theory$ can produce such an output.
Therefore all $M_\theory$ satisfy $\psi_s$.

Let $\psi_t$ be $\globally ( o(q_i) \wedge \updf \rightarrow \lnext \neg o(q_j))$ where $\neg \sat{ o(q_i) \wedge u \wedge  o(q_j)' }$ and $\sat{o(q_i)}$. None of the values satisfying $o(q_i)$ have a successor in $o(q_j)$ after performing $u$.
The added constraint is equivalent to $\neg( o(q_i) \wedge \updf \wedge  o(q_j)' ) \Leftrightarrow \neg o(q_i) \lor \neg \updf \lor  \neg o(q_j)' ) \Leftrightarrow (o(q_i) \land u) \rightarrow \neg o(q_j)' \Leftrightarrow o(q_i) \wedge \updf \rightarrow \lnext \neg o(q_j)$.
All transition in all $M_\theory$ satisfy this property at all points in time.

Let $\psi_p$ be $\globally( \neg p \wedge \updf \rightarrow \lnext \neg o(q_j))$ where $p$ is the weakest precondition of $o(q_j)$ under $u$ and $\sat{ o(q_i) \land u \land \neg o(q_j)}$.
By the definition of weakest precondition no value in $\neg p$ leads to $o(q_j)$ when performing $\updf$.
This also hold in the presence of inputs. The quantifier elimination procedure leads to the weakest precondition for unknown inputs at the next time step.
All transitions in all $M_\theory$ will lead from $\neg p$ to $\neg o(q_j)$ when performing $\updf$.

All added constraints are satisfied by all states and transitions in all $M_\theory$. The constraints only talk about individual states and transitions therefore also all traces in $M_\theory$ satisfy these constraints and $M_\theory \models \psi$.\qed
\end{proof}




\subsection{Generalizing Counterexamples}
\label{sec:opt:general:ce}
The counter examples generated by \cref{alg:consistency:inputs} only block the exact state or transition present in the counter strategy.
To achieve faster and better convergence it is necessary to generalize these counter examples.

Generalization of counter examples is done using an algorithm to find an unsatisfiable core, i.e., a small (not necessarily minimal) subset of clauses such that their conjunction is unsatisfiable. SMT solver such as Z3 \cite{demoura2008} contain an implementation of such a procedure.
We are using unsat cores to find smaller counter examples that do not depend on superficial information. 
Therefore, the counter examples also block situations where unrelated predicates or updates are different.

The counter examples generated in case 1 are straightforward to generalize. A given output predicate, a conjuction of literals, is unsatisfiable by calculating an unsat core we obtain a small set of these literal that is still unsatisfiable.
A counter example $\globally \neg o(q)$ is generalized to $\globally \neg unsatcore(o(q)$.

We can use a similar idea for the assumptions generated in case 2. 
The SMT encoding of a transition is unsatisfiable $\neg \sat{ o(q_i) \land \updf \land  o(q_j)' }$. An unsat core contains subsets of $o(q_i)$, $\updf$, and $o(q_j)'$ that are unsatisfiable, let them be $\overline{o(q_i)}$, $\overline{\updf}$, and $\overline{o(q_j)'}$ respectively. Using these we get the new assumption $\globally ( \overline{o(q_i)} \land \overline{\updf} \rightarrow \lnext \neg \overline{o(q_j)})$

Case 3 seems to be different as it does check for satisfiability of a quantified formula instead of unsatisfiabiliy of a quantifier free conjunction.
However, after finding the new predicates and performing quantifier elimination we get counter examples of a similar structure as in case 2.
The main difference is that the precondition is negated and the unsat core algorithm cannot remove clauses from it.
Because $\rho$ is the weakest precondition we know that the transition is not possible if any of its negated clauses are part of the precondition.
This allows us to split the assumptions into $\bigwedge_l \globally(\neg \rho_l \land \updf \rightarrow \lnext \neg o(q_j))$ where $\rho_l$ are the clauses of $\rho$ and generalize these using the unsat core procedure form case 2.



\section{Synthesis}
\label{sec:synthesis}

We adapt the technique used by Finkbeiner et al. \cite{finkbeiner2019a}, but use a more sophisticated analysis of counter strategies.
The procedure starts with a specification in TSL(T) that is translated to LTL by creating variables for every propositional expression and action.
Every predicate evaluation is treated as an input and only actions are assumed to be outputs.
This is given to a synthesis tool for propositional LTL: if it finds a solution this is also a solution for the TSL(T) synthesis problem,
otherwise a counter strategy is produced.
In case of a counter strategy it is analyzed to find an inconsistency with the used theory and either an extension of the specification is produced, the counter strategy is consistent, or the verifier is unable to decide consistency.
This procedure resembles a CEGAR \cite{clarke2000} loop and is depicted in \cref{fig:overview}.

\begin{algorithm}[H]
 \KwData{TSL(T) modulo theory specification: $\phi$}
 \KwResult{Satisfying Mealy machine or unrealizable or non-termination}
 \While{true}{
     $\phi_\bool$ := prop\_encode($\phi$)\;
     (r,m) := synth($\phi_\bool$)\;
     \eIf{r is \unrealizable}{
       c, $\psi$ := \isconsistent{m}\;
       \eIf{c is \inconsistent}{
        $\phi$ := $\psi \rightarrow \phi$\;
       }{
         return \unrealizable\;
       }
     }
     {
       return m\;
     }
 }
 \caption{Synthesis using abstraction refinement. \label{alg:synth}}
\end{algorithm}

For a initial specification $\phi$ the extension with new assumptions $\psi$
is given as $\phi_n \triangleq \globally \bigwedge_{k=1}^{n} \psi_k \rarrow \phi$ for the $n$-th refinement of $\phi$ with $\psi_k$ the assumptions added in refinement $k$.




\subsection{Illustrative Example}
\label{sec:synthesis:example}
Let us consider the following example: a counter $x$ has to be kept between a minimum and a maximum value.
At each time step the system can choose one of two actions: increase the counter by an input $i$, or decrement the counter.
The environment picks the value of the input variable from a given interval at each time step as well as the initial value of the state variable. We use the concrete values $0 \geq x < 10$ and $0 \geq i < 5$.
Formally this is defined by the TSL(LIA) specification 
$
\phi \triangleq (0 \leq x \land x < 10 \land \globally (0 \leq i \land i < 5)) \rarrow \globally (0 \leq x \land x < 10 \land ([x \leftarrow x-1] \lor [x \leftarrow x+i]))
$.

At first this is encoded as an LTL formula $\phi_\bool \triangleq (p_{0 \leq x} \land p_{x < 10} \land \globally (p_{0 \leq i} \land p_{i < 5})) \rarrow \globally (p_{0 \leq x} \land p{x < 10} \land (p_{[x \leftarrow x-1]} \lor p_{[x \leftarrow x+i]})$ which is given to a propositional synthesis tool. It informs us that $\phi_\bool$ is unrealizible and gives us a counter strategy as an explanation.

\includegraphics[width=0.4\textwidth]{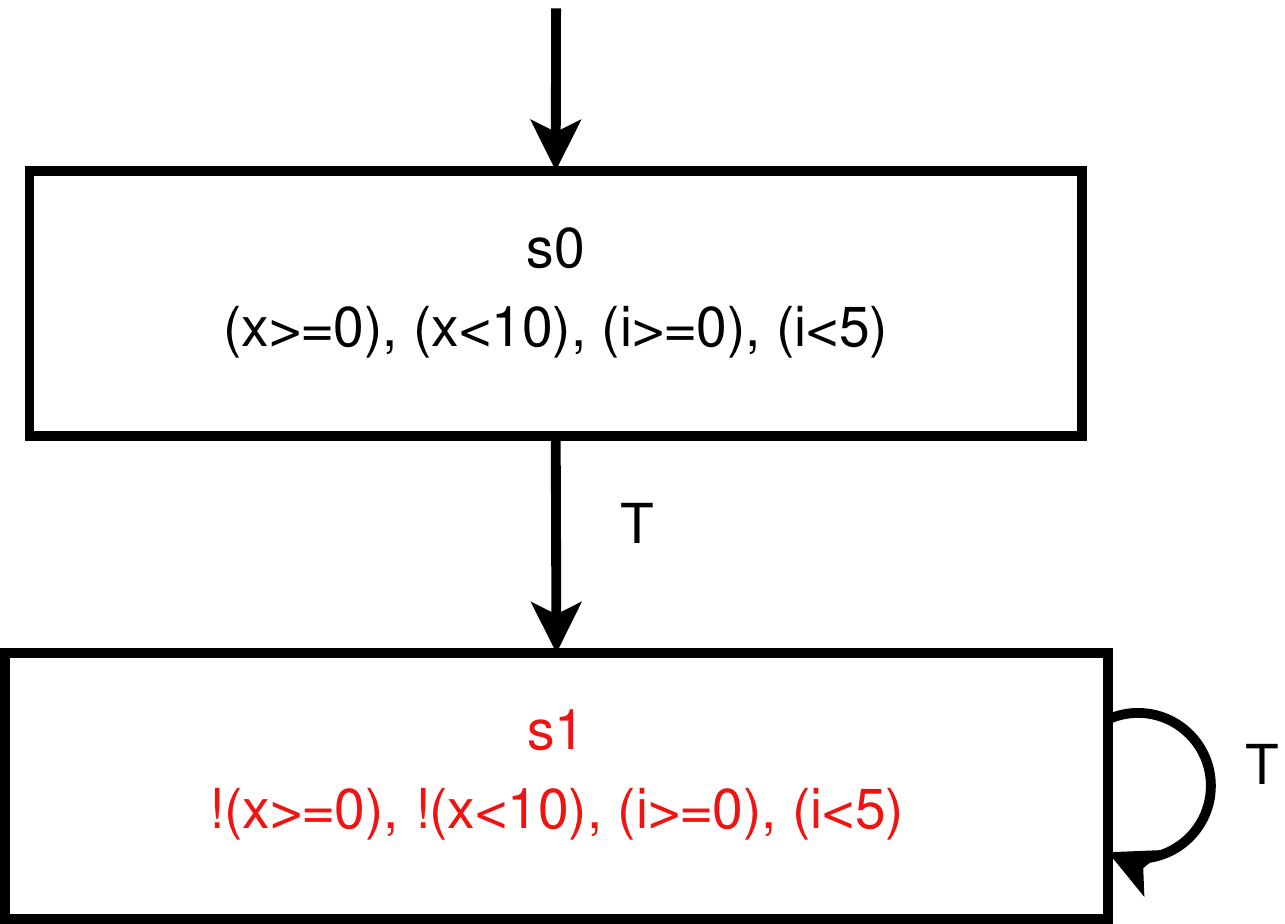}

This counter strategy is analyzed for theory inconsistencies using \cref{alg:consistency:inputs}.
The output of the state $s1$ is inconsistent, because $\neg (0 \leq x ) \land \neg (x < 10)$ is unsatisfiable.
We obtain the new assumption $\psi_1 \triangleq \globally (0 \leq x \lor x < 10)$, note that this is a more general assumption than just the negated state formula. \Cref{sec:opt:general:ce} goes into more detail on how to generalize assumptions. The new assumption is used to extend the specification.
The specification after $r$ iterations is $
\phi_r \triangleq \bigwedge_{k=1}^{r} \psi_k \rarrow ((0 \leq x \land x < 10 \land \globally (0 \leq i \land i < 5)) \rarrow \globally (0 \leq x \land x < 10 \land ([x \leftarrow x-1] \lor [x \leftarrow x+i])))
$ with $\psi_r$ the assumption added in iteration $r$.

Attempting to synthesis a system for $\phi_1$ results in the counter strategy
\includegraphics[width=0.4\textwidth]{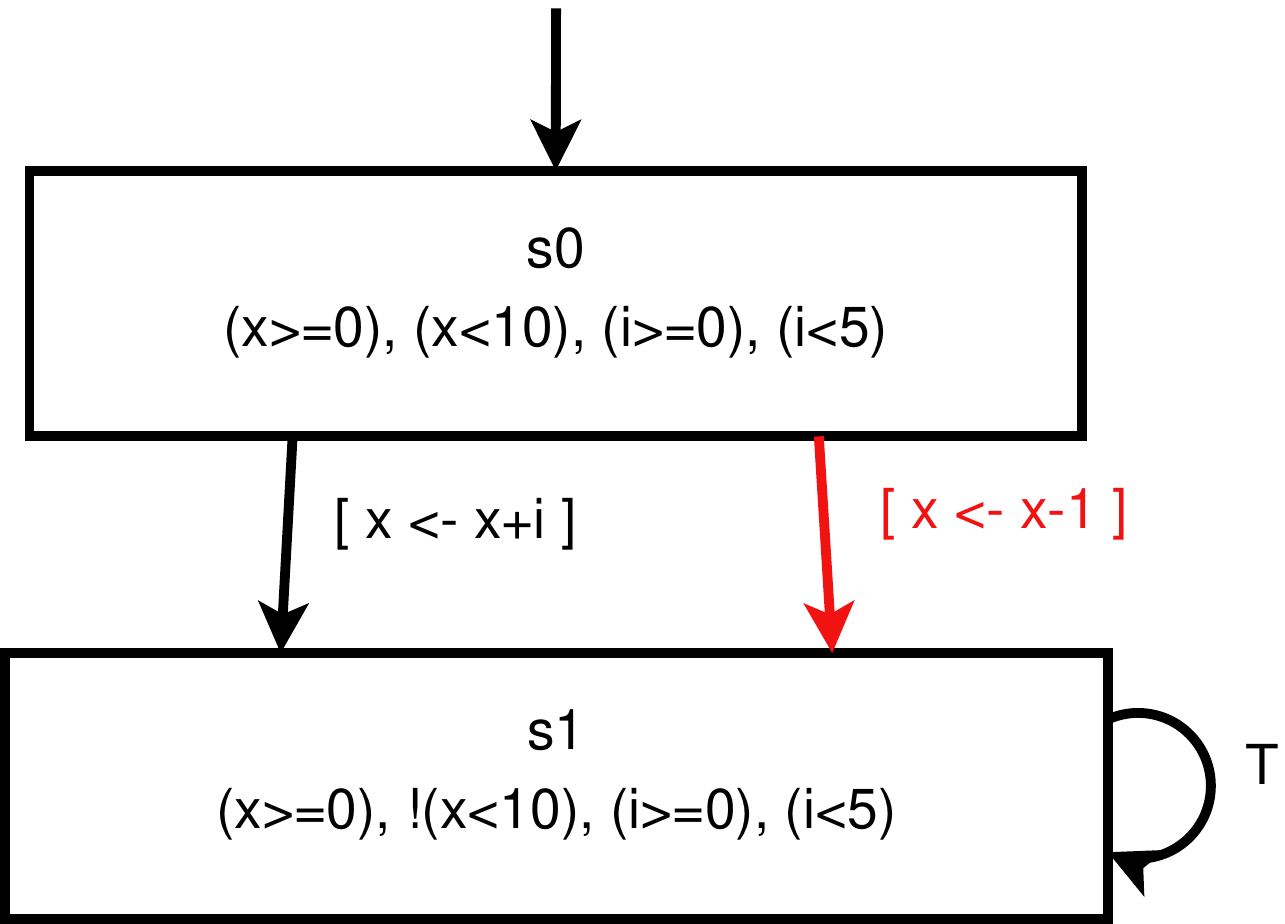}

This time all outputs are consistent, but the transition $(s0,[x \leftarrow x-1],s1)$ is inconsistent.
If $x<10$ and we compute the next value with $[x \leftarrow x-1]$ it cannot be that $x\geq10$ in the next state.
We obtain the assumption $\psi_2 \triangleq \globally (x < 10 \land [x \leftarrow x-1] \rightarrow \lnext x < 10)$.

Boolean synthesis for $\phi_2$ again results in a counter example, this time the transition $(s0,[x \leftarrow x+i],s1)$ is inconsistent.

\includegraphics[width=0.4\textwidth]{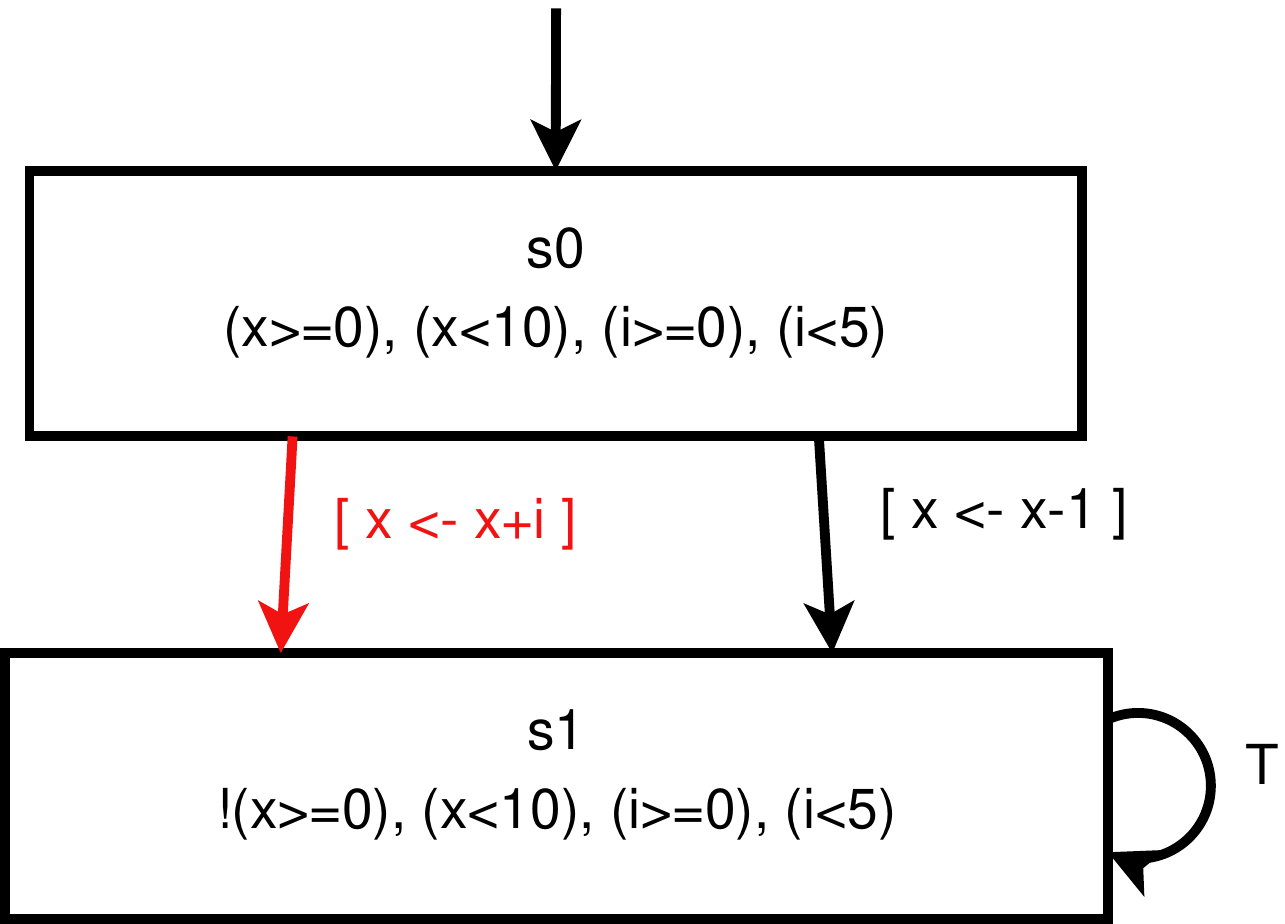}

The generated assumption is $\psi_3 \triangleq \globally (0 \leq x  \land [x \leftarrow x+i] \rightarrow \lnext 0 \leq x )$.

The Boolean synthesis problem for $\phi_3$ is again unrealizible with the counter strategy:

\includegraphics[width=0.9\textwidth]{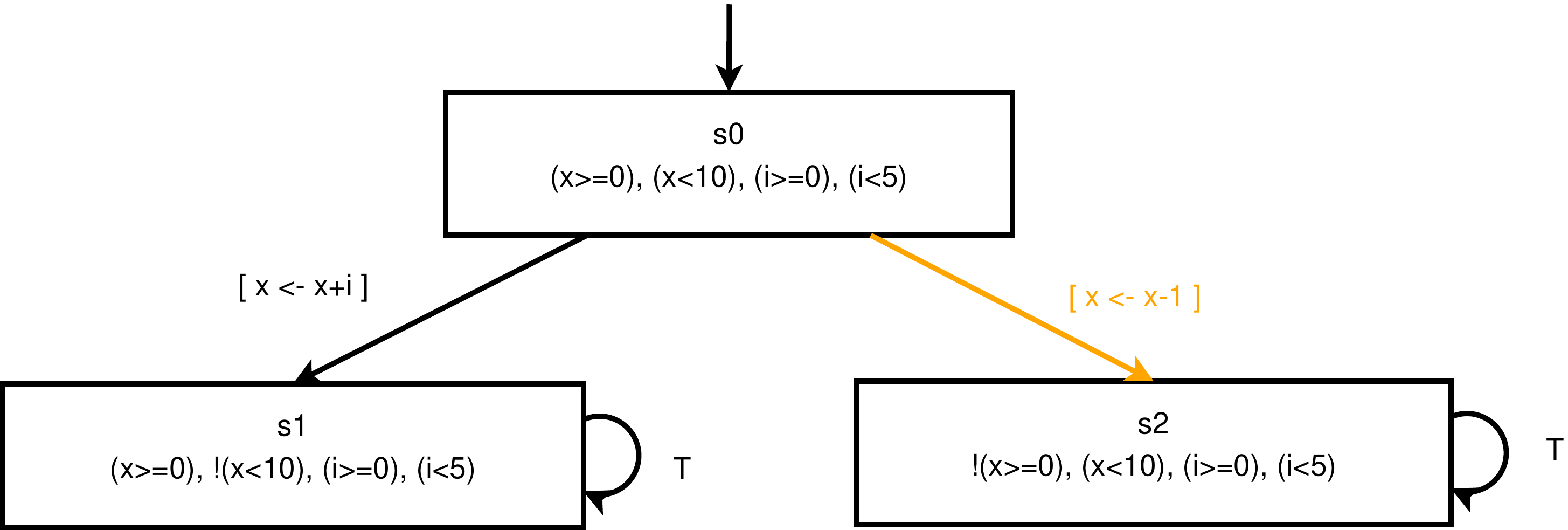}

All the output functions are consistent and none of the transitions are invalid for all valued.
We now have to look at the third category of inconsistencies, transitions that are only valid for some of the concrete states in its abstract origin state. The new predicate $x\geq1$ is learned from analyzing the transition $s0$ to $s2$.
It is part of the assumption $\psi'_4 \triangleq \globally (x \geq 1 \land [x \leftarrow x-1] \rightarrow \lnext 0 \leq x)$.
To reduce the required refinements in this example we immediately add a state consistency assumption and get $\psi_4 \triangleq \globally (x \geq 1 \lor x < 10) \land \psi'_4$.

Running the Boolean synthesis algorithm again results in a counter example:

\includegraphics[width=0.9\textwidth]{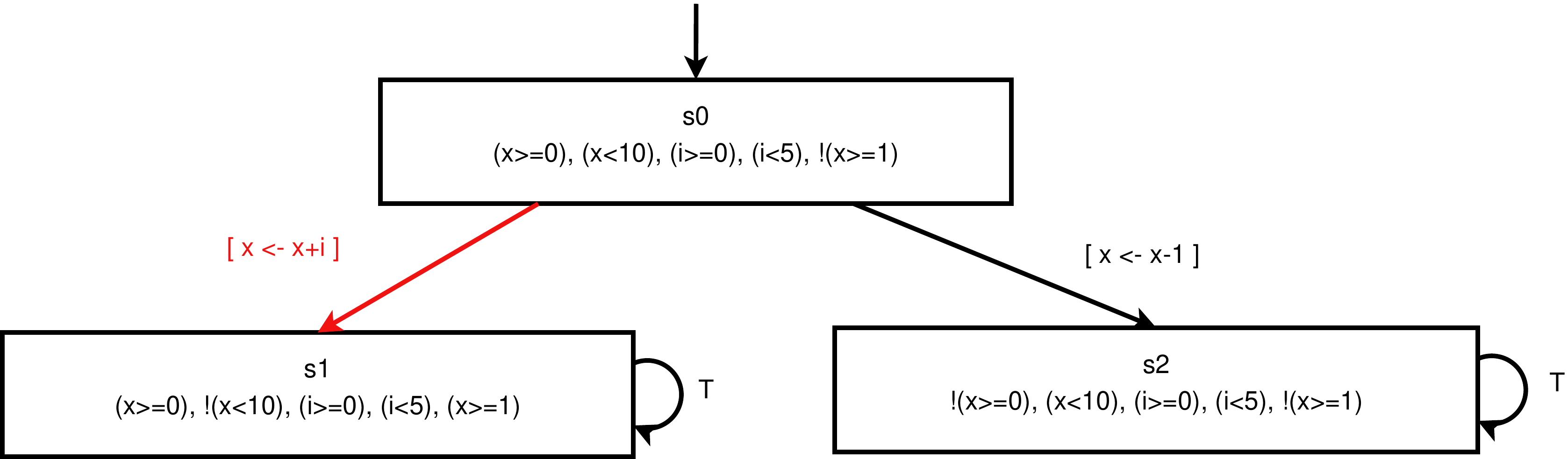}

Transition $(s0,[x \leftarrow x+i],s1)$ is inconsistent. The added assumption is 
$\psi_5 \triangleq \globally (0 \leq x \land x < 1 \land [x \leftarrow x+i] \rightarrow \lnext x < 10)$.

The Boolean synthesis is executed for the last time on $\phi_5$.
This time a Boolean system satisfying the specification is produced.

\includegraphics[width=0.6\textwidth]{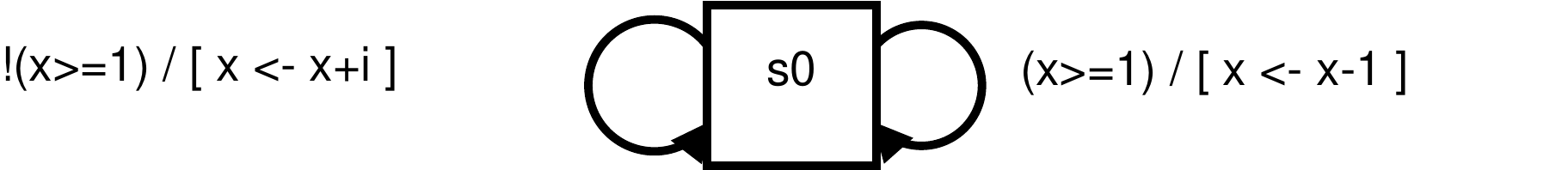}

It can be translated into a concrete imperative program.

\begin{lstlisting}
while(true):
    i := receive()
    if (x>=1):
         x := x - 1
    else:
         x := x + i
    send(x)
\end{lstlisting}

\subsection{Correctness}

\begin{theorem}
If \cref{alg:synth} returns unrealizible there is no $M_\theory \models \phi$.
\end{theorem}
\begin{proof}
If there exists a machine $M_\theory' \models \neg \phi$ there is no machine $M_\theory \models \phi$.
The propositional synthesis tool provides us with a machine $M_\bool \models \neg \phi$.
The consistency check results in consistent, so by \cref{lem:consistent} there exists a $M_\theory'\models \neg \phi$.
\qed
\end{proof}

\begin{lemma}
If $M_\bool$ satisfies a propositionally encoded formula $\phi_\bool$ than any of its conrcretizations $M_\theory \sim M_\bool$ satisfy the TSL(T) formula $\phi$.
$$M_\bool \models \phi^\bool \Rightarrow \forall M_\theory \sim M_\bool \ldotp M_\theory \models \phi$$
\label{lem:BoolRealImpliesTheoryReal}
\end{lemma}
\begin{proof}
$M_\bool$ wins against all interpretations of the predicates. 
The theories semantics impose one interpretation of the theory symbols.

So $M_\theory$ wins against the one interpretation imposed by the theories semantics.\qed
\end{proof}

\begin{theorem}
If \cref{alg:synth} returns a system $M_\theory$ it holds that $M_\theory \models \phi$.
\end{theorem}
\begin{proof}
The \cref{alg:synth} returns a system $M_\bool$ that for some refinement $r$ satisfies the Boolean specification $\phi_r^\bool$. 
By \cref{lem:BoolRealImpliesTheoryReal} any concretization $M_\theory \sim M_\bool$ satisfies $\phi_r$.
From \cref{lem:tautogicalAssumptions} it follows that $M_\theory$ satisfying $\phi_r \equiv (\bigwedge_k^r \psi_k) \rightarrow \phi$ satisfies $\phi$.\qed
\end{proof}






Even though our algorithm is not guaranteed to terminate it can proof unrealizabilty in certain cases.
$$
x=0 \rarrow \globally ( [x \leftarrow x+1] \land x < 3 )
$$
Here we can perform two refinement steps and learn the new predicates $x \geq 2$ and $x \geq 1$.
Using these the propositional synthesis tool is able to build a consistent environment strategy.
There are no conflicts which could be used to further refine the specification.
This shows that the specification is unrealizable.

\subsection{Limitations}
\label{sec:synth:limitations}
\begin{theorem}
The synthesis problem for TSL(T) modulo theories is undecidable.
\end{theorem}

\begin{proof}
A TSL(T) formula with three state variables and a theory supporting: unbounded integers, increment, decrement, and equality, can simulate a two counter machine. We consider a two counter machine \cite{minsky1967} with the three instructions increment register $inc(r)$, decrement register $dec(r)$, and jump if zero $jz(r,z)$.
A program for a two counter machine is a list of instructions, with out loss of generality we assume every instruction is prefixed with its position in the instruction list (the label l) the last element of the list is empty and has the label $h$. We also assume that all jump addresses point to an instruction in the list.
The instructions are encoded as:
\begin{align*}
\llbracket l:inc(R1) \rrbracket &\triangleq ic = l \land [ic \leftarrow ic + 1] \land [r1 \leftarrow r1 + 1] \land [r2 \leftarrow r2]\\
\llbracket l:dec(R1) \rrbracket &\triangleq ic = l \land [ic \leftarrow ic + 1] \land [r1 \leftarrow r1 - 1] \land [r2 \leftarrow r2]\\
\llbracket l:jz(R1,z) \rrbracket &\triangleq (ic = l \land r1=0 \land [ic \leftarrow z] \land [r1 \leftarrow r1] \land [r2 \leftarrow r2])\\ 
&\lor (ic = l \land r1\neq0 \land [ic \leftarrow ic +1] \land [r1 \leftarrow r1] \land [r2 \leftarrow r2])
\end{align*}
the instruction for $R2$ are the same except $r1$ and $r2$ are flipped.
Given this encoding we can build the formula
$
(ic = 0 \land r1 = 0 \land r2 = 0) \rightarrow ( \bigvee_{instr \in PROG} \llbracket instr \rrbracket \land \eventually ic = h )
$
for a program $PROG$.

This is realizable iff the program holds. 
The halting problem for two counter machines is undecidable.
Therefore the synthesis problem for TSL(T) mod theories is undecidable.
This also applies to formulas with only only one state variable.
In that case one can construct a theory over the domain $\mathbb{N}\times\mathbb{N}\times\mathbb{N}$ with operations that work on the individual fields.\qed
\end{proof}

Our algorithm cannot handle reachability properties where the number of required steps depends on the concrete value of a state variable and is unbounded. The specification 
$$0 \leq x \rarrow (\eventually (x<0) \land \globally ( [x \leftarrow x+1] \lor [x \leftarrow x-1] ))$$
with the state variable $x$ is an example of this happening.
The specification is obviously realized by a system always using the update $[x \leftarrow x-1]$.
However, we would add the new predicates $x\geq1, x\geq2, \ldots$ with out ever terminating.

The algorithm fails to prove unrealizability in certain cases.
$$
x=1 \rarrow (\eventually (x=0) \land \globally [x \leftarrow x+1])
$$
Here we would learn the predicates $x=-1, x=-2, \ldots$ with out terminating.
Note that knowing the predicate $x>1$ would allow us to prove unrealizability.
It can serve as an invariant of the loop in a bad lasso.


\section{Experimental Evaluation}
\label{sec:experiments}
We implemented our algorithm in Haskell.
Our implementation relies on several external tools: tsltools \cite{finkbeiner2019a} is used for parsing TSL and to perform the propositional encoding, strix \cite{meyer2018} is used for LTL synthesis, and Z3 \cite{demoura2008} is used as the SMT solver.
When performing counter example analysis using \cref{alg:consistency:inputs} we add all assumptions from the same case before we start the next iteration. Once a system has been found we rerun strix on the last refinement with the option to minimize the number of states. This allows us to obtain a more compact system.

\subsection{Illustrative example extended}
The first experiment is an extension of the illustrative example from \cref{sec:synthesis:example}.
The system is no longer allowed to change between the two updates at every step.
Instead after changing the update it has to use the new update for the next $c$ steps.
We also varied the size of the intervals for $x$ and $i$ demonstrating that our algorithm is independent of the size of the concrete state space. The results table lists the used parameters ($c$, $x_{max}$, $i_{max}$), the number of refinments, the number of state in the minimized system, the number of learned predicates during the whole execution and the total runtime in seconds.

\begin{align*}
&0 \leq x \land x \leq x_{max} \land \globally (i \geq 0 \land i \leq i_{max}) \rarrow \globally (\\
&0 \leq x \land x \leq x_{max} \land ([x \leftarrow x-i] \lor [x \leftarrow x+i]) \land\\
&(([x \leftarrow x-i] \land \lnext [x \leftarrow x+i]) \rightarrow \globally_{[1,c]} [x \leftarrow x+i]) \land\\
&(([x \leftarrow x+i] \land \lnext [x \leftarrow x-i]) \rightarrow \globally_{[1,c]} [x \leftarrow x-i]))
\end{align*}

\begin{center}
\begin{tabular}{ |c|c|c|c|c|c|c| } 
\hline
$c$ & $x_{max}$ & $i_{max}$ & \# refinements & \# states & \# learned predicates & time [s]\\
\hline
1 & 100 & 5 & 4 & 1 & 2 & 1.0 \\
2 & 100 & 5 & 5 & 2 & 2 & 1.3 \\ 
2 & 1 000 & 5 & 5 & 2 & 2 & 1.3 \\ 
2 & 100 000 & 50 & 5 & 2 & 2 & 1.3 \\ 
3 & 100 & 5 & 9 & 2 & 4 & 2.9 \\ 
3 & 1 000 & 5 & 9 & 2 & 4 & 2.9 \\ 
3 & 100 000 & 50 & 10 & 2 & 4 & 3.1 \\ 
3 & 1 000 000 & 5 000 & 9 & 2 & 4 & 3.0 \\ 
\hline
\end{tabular}
\end{center}

\subsection{Elevator}
A classic example for reactive synthesis is a controller for an elevator.
We include two variants: the first naive one has no inputs and needs to visit all floors infinitely often, the second version contains an input signal where a user can tell the elevator where it should go.

\paragraph{Naive Elevator} The single state variable $floor$ represents the current position of the elevator.
It can start anywhere between the first floor and the maximum floor and is not allowed to leave this interval.
The controller has three options: move elevator up or down or stay at the same position.
Every floor has to be visited infinitely often. This can be expressed as the TSL(LIA) formula

\begin{align*}
&(floor \geq 1 \land floor \leq max) \rarrow \globally (\\
&([floor \leftarrow floor] \lor [floor \leftarrow floor - 1] \lor [floor \leftarrow floor + 1]) \land \\
&floor \geq 1 \land floor \leq max \land \bigwedge_k^{max} \eventually floor = k)
\end{align*}
We varied the number of floors of the building to show how our algorithm scales with more complex specifications.
No new predicates are learned as a sufficient number of predicates is already included in the specification (equality tests for every floor are part of the liveness properties).
\begin{center}
\begin{tabular}{ |c|c|c|c| } 
\hline
\# floors & \# refinements & \# states & time [s]\\
\hline
3  & 13 & 2 & 3.1 \\
4  & 11 & 3 & 3.7 \\ 
5  & 15 & 4 & 8.2 \\ 
8  & 21 & 4 & 45 \\ 
10 & 24 & 4 & 185 \\ 
\hline
\end{tabular}
\end{center}

\paragraph{Elevator with an input signal}
The state consists of two variables: the current $floor$, and the $target$ we want to reach.
There is one input variable $signal$ used by the environment to request the next target.
We assume the floors are numbered $1$ to $max$ and the number $0$ denotes an undefined value.
As in the naive elevator example the elevator can move up or down or stay at the same level. 
It must stay between the minimum and the maximum floor.
If there is currently no target the $signal$ can be any valid floor and this will be stored as the new $target$.
Whenever a new target is selected the elevator has to eventually reach it, at that point $target$ is reset to zero.

\begin{align*}
&(floor \geq 1 \land floor \leq max \land target = 0 \land\\
&\globally (signal \geq 0 \land signal \leq max \land (target \neq 0 \rightarrow signal = 0))) \rarrow \globally (\\
&([floor \leftarrow floor] \lor [floor \leftarrow floor - 1] \lor [floor \leftarrow floor + 1]) \land \\
&((signal \neq 0 \land floor \neq target) \rightarrow [target \leftarrow signal]) \land\\
&((signal = 0 \land floor \neq target) \rightarrow [target \leftarrow target]) \land\\
&(floor = target \rightarrow [target \leftarrow 0]) \land\\
&floor \geq 1 \land floor \leq max \land \bigwedge_k^{max} (target=k \rightarrow \eventually floor = k))
\end{align*}

\begin{center}
\begin{tabular}{ |c|c|c|c| } 
\hline
\# floors & \# refinements & \# states & time [min]\\
\hline
3  & 5 & 1 & 0:35 \\
4  & 5 & 1 & 3:37 \\ 
5  & 6 & 1 & 23:44 \\ 
\hline
\end{tabular}
\end{center}

The specification above explicitly enumerates all floors and contains a liveness guarantee for each of them.
Instead one could only require that $target$ is reset infinitely often.
This makes the problem significantly harder: 59min and 14 refinements for only three floors.

\subsection{Sorting}
TSL(LIA) can be used to specify a sorting algorithm for a finite number of elements.
Every element is represented by its own variable, which are initialized by the environment to arbitrary integers.
The system is allowed to exchange two adjacent number or leave everything unchanged.
It has to guarantee that eventually the numbers are sorted and stay that way.
The TSL(T) formula for three variables is:

\begin{align*}
&\eventually \globally (a \geq b \land b \geq c) \land \globally(\\
&([a \leftarrow b] \land [b \leftarrow a] \land [c \leftarrow c]) \lor\\
&([a \leftarrow a] \land [b \leftarrow c] \land [c \leftarrow b]) \lor \\
&([a \leftarrow a] \land [b \leftarrow b] \land [c \leftarrow c]))
\end{align*}

This problem turns out to be very hard, while a system for three numbers can be build in less than 4s, four numbers took almost 20min and five numbers timed out after 13h.

\begin{center}
\begin{tabular}{ |c|c|c|c| } 
\hline
\# variables & \# refinements & \# states & time\\
\hline
3  & 6 & 2 & 3.7 s \\
4  & 12 & 3 & 18:46 min \\ 
5  & $>$10 & - & timeout ($>$13 h) \\ 
\hline
\end{tabular}
\end{center}

\subsection{Water Tanks}
The previous examples all used linear integer arithmetic.
We can also use other SMT theories like linear real arithmetic (LRA).
Using reals allows us to model linear cyber physical systems.

Belta et al. \cite{belta2017} chapter 9 describes a system of two coupled water tanks with linear dynamics; one water tank drains ($x2$) and the other one ($x1$) is refilled by the controller. The original input domain (refill tank x1) is $i \in [0,0.0005]$.
The input choice is discretized with two values ($0$ and $0.0003$) and represented as different updates.
We created two variants of the system.

The first one is a safety specification where the water level of both tanks has to be kept between $0.1$ and $0.7$.
\begin{lstlisting}[mathescape=true]
$(0.2 \leq x1 \land x1 < 0.7 \land 0.1 \leq x2 \land x2 < 0.7)$
$\rarrow \globally ($
$0.1 \leq x1  \land x1 < 0.7 \land 0.1 \leq x2 \land x2 < 0.7 \land$
$((x1 < 0.2 \land x2 < 0.2) \rightarrow$
  $(([x1 \leftarrow x1 + 0.0   *324.6753] \lor [x1 \leftarrow x1 + 0.0003*324.6753]) \land$
  $[x2 \leftarrow 0.9635*x2])) \land$
$((x1 \geq 0.2 \lor x2 \geq 0.2) \rightarrow$
  $(([x1 \leftarrow 0.8281*x1 + 0.1719*x2 + 0.0   *324.6753] \lor$
    $\!\![x1 \leftarrow 0.8281*x1 + 0.1719*x2 + 0.0003*324.6753]) \land$
  $[x2 \leftarrow 0.7916*x2 + 0.1719*x1]))$
$)$
\end{lstlisting}

Synthesis of this system took 31 seconds and 4 refinements. It results in a system with a single state.

A specification with two tanks and a liveness property is currently out of reach for our tool.
Instead we created a second specification with a single tank and a liveness property.
Whenever the water level falls below $0.1$ it has to eventually exceed $0.4$.





\begin{align*}
&0.0 \leq x \land x < 0.7 \rarrow \globally (\\
&([x \leftarrow 0.9635*x] \lor [x \leftarrow 0.9635*x + 0.1]) \land\\
&(x<0.1 \rightarrow \eventually (x>0.4)))
\end{align*}

A system realizing that specification can be synthesized using 18 refinements in 95 seconds, it consists of 2 states.

\section{Conclusion}
\label{sec:conclusion}
We presented a synthesis procedure for temporal stream logic modulo theories.
Our algorithm is based on a CEGAR \cite{clarke2000} loop and translation to propositional LTL synthesis.
The synthesis problem for TSL modulo theories in general is undecidable.
However, we can synthesize systems or prove unrealizability in many cases.
Huge state spaces can be handled by using a symbolic representation during synthesis.
Some specifications require a system with predicates that are not part of the original specification, in many cases we are able to automatically find these.


%
%
\bibliographystyle{splncs04}
\bibliography{mybibliography}
\end{document}